\documentclass[11pt]{scrartcl}
\usepackage[a4paper, total={16cm, 24cm}]{geometry}

\usepackage{multicol}

\usepackage[small]{caption}
\usepackage{subcaption}

\usepackage[table,dvipsnames]{xcolor}
\usepackage{todonotes}
\usepackage{pdflscape}
\usepackage{adjustbox}

\usepackage{natbib}
\usepackage[pagebackref]{hyperref}
\hypersetup{
	pdfencoding=auto,
	psdextra,
	colorlinks=true,
	citecolor=green!50!black,
	linkcolor=red!60!black,
	urlcolor=purple!70!black
}

\renewcommand{\epsilon}{\varepsilon}

\usepackage{authblk}

\author[1]{Benjamin Cookson}
\author[2]{Eva Deltl}
\author[3]{Yeeseok Oh}
\affil[1]{University of Toronto, \texttt{bcookson@cs.toronto.edu}}
\affil[2]{TU Clausthal, \texttt{eva.deltl@tu-clausthal.de}}
\affil[3]{The University of Tokyo, \texttt{oh@g.ecc.u-tokyo.ac.jp}}

\usepackage{url}
\usepackage{microtype}
\usepackage[utf8]{inputenc}
\usepackage{graphicx}
\usepackage{amsmath}
\usepackage{amsthm}
\usepackage{amssymb}
\usepackage{mathtools}
\usepackage{nicefrac}
\usepackage{booktabs}
\usepackage{enumitem}
\usepackage{tabulary}
\usepackage{tabularx}
\usepackage{dsfont}
\usepackage{algorithm}
\usepackage{algpseudocodex}

\usepackage{wrapfig}

\usepackage[nameinlink]{cleveref}
\usepackage{stmaryrd}
\usepackage{thm-restate}
\usepackage{tikz}
\usetikzlibrary{arrows.meta}
\usepackage{tcolorbox}

\newtheoremstyle{sfthm}
	{\topsep}
	{\topsep}
	{\itshape}
	{}
	{\sffamily\bfseries}
	{}
	{.5em}
	{}
\theoremstyle{sfthm}
\newtheorem{theorem}{Theorem}[section]
\newtheorem{lemma}[theorem]{Lemma}

\newtheoremstyle{sfdef}
{\topsep}
{\topsep}
{}
{}
{\sffamily\bfseries}
{}
{.5em}
{}
\theoremstyle{sfdef}

\title{Optimally Selecting Representative Agents from a Metric Space}

\date{}

\begin{document}

\maketitle

\begin{abstract}
\begin{center}
	\textbf{\textsf{Abstract}} \smallskip
\end{center}
This paper studies the problem of proportionally fair clustering, where the goal is to select $k$ ``centers'' from a metric space that fairly represent a set of agents who also lie in the metric space. Specifically, we focus on finding a clustering satisfying a fairness property known as the Droop core. In the practical special case in which the set of feasible center locations contains every agent location, the previous best-known result guaranteed a $(1 + \sqrt{2})$-approximation of the Droop core, while the best-known lower bound was $2$. In this paper, we show that this lower bound is tight and that a clustering in the $2$-Droop core always exists. Further, we show that such a clustering can be achieved by only selecting centers from locations in the metric space where an agent resides. We establish this using Scarf's theorem guaranteeing a nonempty core for balanced non-transferable utility games. This result has several interesting corollaries. Most notably, it resolves the $\beta$-plurality problem of \citet{aronov2020beta} for general metric spaces. The main result of this paper was generated by $\mathtt{ChatGPT}$-$\mathtt{5.6}$-$\mathtt{Sol}$ through a series of interactions with the authors. The authors of this paper verified the generated proof and rewrote it for clarity.
\end{abstract}

\section{Introduction}

One of the most fundamental problems in social choice theory is how to select a committee of $k$ candidates from a set $C$ to best represent a set of agents $N$ who have preferences over those candidates. These preferences can take many forms, such as ordinal rankings, cardinal values, or approval sets. In one interesting preference model, the agents and candidates both reside at points in some underlying metric space, and an agent's preference for a candidate depends on the metric distance between them.

The problem of selecting a fair committee under these metric preferences is typically known as a \emph{fair clustering} problem. Selecting candidates from a metric space so that every agent has a nearby candidate closely resembles the traditional problem of clustering from machine learning, where a number of \emph{centers} must be chosen from a metric space to minimize an objective function based on the distances between those centers and various data points in the space. \citet{ChenEtal2019} originally introduced this problem and the clustering terminology, referring to the candidates in the metric space as centers. They were the first to observe that selecting candidates in this setting using classical clustering objectives such as $k$-means and $k$-medians did not produce any compelling fairness guarantees.

Specifically, \citet{ChenEtal2019} were interested in finding a set of $k$ centers that was in \emph{the $\alpha$-core} with respect to the agents. A set of candidates $X \subseteq C$ is in the $\alpha$-core with respect to a set of agents $N$ if no group of agents that is big enough to ``deserve'' at least one cluster center to represent them (\citet{ChenEtal2019} set this group size to $\nicefrac{|N|}{k}$) can come together and select some new center that is $\alpha$ times closer to each of them than any center in $X$. A clustering satisfying this notion is said to achieve a form of proportional representation. No large group of agents can claim that they are \emph{under-represented} by the chosen clustering by providing a mutually preferred option.


In their original paper, \citet{ChenEtal2019} showed that every proportionally fair clustering instance admits a $(1 + \sqrt{2})$-core clustering. They also provided a lower bound showing that no algorithm can achieve a better approximation to the core than $2$. Recently, \citet{cookson2026improved} improved the lower bound to $2.1508$, removing the possibility that a clustering algorithm achieving a $2$-approximation of the core may exist.


However, this did not eliminate such a possibility for many of the most practical applications of proportionally fair clustering. In many prominent applications, it makes sense to assume that the set of valid centers is a \emph{superset} of the set of points where an agent resides. This is not the case in the lower-bound instances found by \citet{cookson2026improved}. This restriction captures many natural problems that use the proportionally fair clustering model. For instance, in many problems, the set of valid cluster centers is \emph{exactly} the set of locations where agents reside, modeling settings in which a small subset of a population of agents must be selected to fairly represent the entire population. Examples include the metric sortition problem studied by \citet{ebadian2024boosting} and the problem of selecting representative slates of user-generated opinions studied by \citet{kraiczy2026social}.
More generally, any problem in which every point in the underlying metric space is a valid center will have this property.

We refer to this special case as the ``$N \subseteq C$'' case, indicating that any agent (or, more specifically, the point where that agent resides) can be selected as a center. In this case, the best-known upper and lower bounds were identical to those originally found for the general case, $[2, 1 + \sqrt{2}]$, with the upper bound achieved by the same Greedy Capture algorithm that works in the general case and the lower bound coming from an instance where $N = C$ (the set of valid centers was exactly the set of points where agents reside) \citep{Trinh2024, cookson2026improved}.

\subsection{Our Contributions}
We settle this question for the $N\subseteq C$ case by showing that the lower bound of $2$ is tight. In fact, we show that such a clustering can always be chosen using \emph{only} points where the agents reside. Further, we show that this clustering will be in the $2$-core even when every group of agents of size greater than $\frac{|N|}{k+1}$ is able to deviate (known as the Droop quota), a strictly stronger guarantee than the original definition of the core, which only allows groups of size at least $\frac{|N|}{k}$ to deviate (the Hare quota).

\begin{theorem}\label{thm:2-Droop-core}
For every clustering instance with $N \subseteq C$, a clustering $X \subseteq N$ in the $2$-Droop core exists.
\end{theorem}

In addition to fully resolving the bounds for the $\alpha$-Droop core in the special $N \subseteq C$ case, our results have several interesting corollaries. In particular, our result resolves the $\beta$-plurality problem \citep{aronov2020beta}, which is precisely the problem of finding an approximate Droop core clustering when $k=1$ (see \citep{kellerhals2026proportional} for details).

The proof proceeds by first defining a notion of \emph{fractional} clusterings and a fractional version of the core. We use a classical theorem by \citet{scarf1967core} to show that for any clustering instance, a fractional clustering in the (1-)core always exists. We then use a simple greedy algorithm to turn this fractional solution into an integral clustering in the $2$-Droop core.

The proof of \Cref{thm:2-Droop-core} and the proofs of all supporting lemmas were fully generated by $\mathtt{ChatGPT}$-$\mathtt{5.6}$-$\mathtt{Sol}$. Specifically, when the authors first asked the model to solve the problem in \Cref{thm:2-Droop-core}, it was not able to produce an answer, even after many rounds of thinking. However, when prompted instead to solve the $\beta$-plurality problem (the special case of \Cref{thm:2-Droop-core} when $k = 1$), the model produced a simple proof showing $2$-Droop core existence in the $k=1$ regime. This proof relied on setting up and solving a linear program, whose optimal solution could be rounded into the selection of a single center achieving the $2$-Droop core. The model was then prompted again to solve the problem for general $k$, this time together with this new solution for $k=1$ and a query to consider a similar strategy of finding a nice object that resembles exact-core and then rounding to a $2$-Droop core. This time, the model was able to find a fractional solution via Scarf's theorem, leading to the proof that appears in this paper. The authors have heavily rewritten this proof for clarity, but the core arguments remain the same.

\subsection{Related Work}\label{sec:related-work}

After \citet{ChenEtal2019} introduced proportional fairness for centroid clustering and showed that their \emph{Greedy Capture} algorithm always obtains a $(1+\sqrt{2})$-core, subsequent work identified settings in which stronger guarantees are possible. In particular, \citet{MichaShah2020} showed that Greedy Capture achieves a factor of $2$ in Euclidean spaces under the $\ell_2$ norm, and that an exact core clustering exists for tree metrics.
Recently, \citet{cookson2026improved} improved the lower bound for general metric spaces to $2.1508$. Their construction, however, relies on candidate sets that do not include all the agent locations. In contrast, we focus on the natural setting $N\subseteq C$, where every agent is an admissible center, and show that the previously known lower bound of $2$ is tight.

Several recent works place proportional clustering in the broader context of proportional representation. \citet{KellerhalsPeters2024} relate proportional clustering to notions from multiwinner voting and to individual fairness in clustering, and introduce metric analogues of justified-representation axioms. \citet{AzizEtal2024} propose \emph{Proportionally Representative Fairness} (PRF), which requires large and cohesive groups to receive a proportional number of nearby representatives. The core and these other notions of fairness in clustering have inspired applications such as representative committee selection and sortition \citep{ebadian2024boosting}, as well as fair facility and transit-stop placement \citep{aziz2026fair}.

The single-center case has also been studied independently through the $\beta$-plurality problem. \citet{aronov2020beta} introduced $\beta$-plurality points and established tight or near-tight guarantees for several Euclidean settings. \citet{filtser2020plurality} extended this study to arbitrary metric spaces, obtaining the same bounds of $[2,1+\sqrt{2}]$ that hold for the setting of general $k$. More recently, \citet{kellerhals2026proportional} observed that $\beta$-plurality is precisely the $k=1$ case of proportional clustering under the Droop quota, and further investigated this connection in the ordinal model.

Proportionality has also been studied for clustering models in which agents care about the other members of their cluster rather than only about a selected center. \citet{caragiannis2024proportional} initiated the systematic study of proportional fairness in the \emph{non-centroid} setting, considering both core stability and the weaker notion of fully justified representation. When each agent's loss is based on the agent's distance to the farthest other agent in the same cluster, they obtain constant-factor guarantees by adapting ideas from Greedy Capture. \citet{bredereck2025core} subsequently provided improved lower bounds for core approximations in this model. \citet{cookson2025unifying} proposed a framework encompassing both centroid and non-centroid clustering, clarifying which guarantees and algorithmic ideas extend between the two settings.

Our proof uses a different approach from the geometric and greedy arguments underlying most previous results on proportional clustering. It relies on the seminal theorem of \citet{scarf1967core}, which guarantees a nonempty core for balanced non-transferable-utility games. Scarf's theorem, as well as the combinatorial lemma underlying it, have previously been used to establish fractional stability in several discrete coalition-formation settings. In particular, \citet{aharoni2003} use Scarf's lemma to establish the existence of fractional stable matchings in hypergraphs, while \citet{biro2016} develop a fractional-core interpretation for capacitated NTU games and apply it to stable matching problems. Our argument follows a related philosophy by first establishing exact core stability for a fractional clustering and then exploiting the metric structure to round it to an integral solution, losing only a factor of~$2$.

\section{Preliminaries}\label{sec:prelims}
For a positive integer $t$, denote $[t] = \{1,2, \ldots, t\}$.
For a vector $\mathbf{z} \in \mathbb{R}^E$ indexed by a finite set $E$ and a subset $A \subseteq E$, let $\mathbf{z}(A) = \sum_{a \in A}z_a$.
For a finite set $A$ and a nonnegative integer $r$, let $\binom{A}{r}=\{B\subseteq A:|B|=r\}$.

\paragraph{Clustering instances.}
A traditional \emph{clustering instance} $I=(N,C,d,k)$ consists of a set of $n$ \emph{agents} $N = [n]$, a set of \emph{centers} $C$, and a \emph{pseudometric} $d\colon(N\cup C)\times(N\cup C)\to\mathbb{R}_{\ge 0}$ satisfying the following metric axioms for all $x,y,z \in N \cup C$:
\begin{itemize}
	\item \textbf{Reflexivity:} $d(x,x) = 0$
	\item \textbf{Symmetry:} $d(x,y) = d(y,x)$
	\item \textbf{Triangle inequality:} $d(x,z) \leq d(x,y) + d(y,z)$
\end{itemize}
\noindent
The instance also includes a positive integer $k\le |C|$ specifying the maximum number of centers that may be \emph{chosen}.

A \emph{clustering} is a subset $X\subseteq C$ with $1 \leq |X| \leq k$. For agent $i\in N$, the \emph{loss} under $X$ is $\ell_i(X)=\min_{c\in X}d(i,c)$ ($i$'s distance to the closest center in the clustering).

We will augment the traditional clustering instance by primarily considering \emph{weighted clustering instances}, used by \citet{cookson2026improved} in the clustering setting. A weighted clustering instance takes a traditional clustering instance and adds a weight function $w\colon N\to\mathbb{Q}_{>0}$ such that $\sum_{i \in N}w(i) = 1$. Further, for any $S \subseteq N$, we write $w(S) = \sum_{i \in S}w(i)$. A weighted instance can be thought of as representing a standard clustering instance where multiple agents may reside at the same point. The weight $w(i)$ of an agent $i \in N$ represents the total fraction of the agent population residing at that point. A weighted instance can be translated back into an unweighted one by scaling to a common denominator and adding the appropriate number of copies according to the weight of each agent. Weighted instances not only provide a more compact representation but also allow each point in the metric space to be represented by a unique element, which is particularly helpful for exposition. In the weighted metric setting, we can, without loss of generality, assume that our distance metric $d$ adheres to a strengthened version of the reflexivity axiom:
\begin{itemize}
	\item \textbf{Identity of indiscernibles:} $d(x,y) = 0 \iff x = y$
\end{itemize}

We no longer need multiple copies of each agent at a single point in the space, since we can simply represent the same information with a single agent who has a greater weight. Notably, this also means that if there is some $i \in N$ and $c \in C$ such that $d(i,c) = 0$, then it must be the case that $i = c$; that is, they must be the same element\footnote{Note that this also implies that for any $c,c' \in C$ with $d(c,c') = 0$, we must have $c = c'$. The traditional clustering setting allows for multiple potential centers to be located at the same place in the metric space, and for the returned clustering to have multiple such centers. Condensing all such centers into a single element in our model is without loss of generality for the existence of the fairness notions we consider. With loss functions that rely on an agent's distance to their closest center, it is never beneficial to select two centers in the exact same location.}. This is intentional. We also allow elements representing agents to serve as centers. This allows us to characterize instances as having $N \subseteq C$ or $N = C$, among other things, making exposition much simpler. We adopt weighted instances as the default throughout this paper and write them as $I=(N,C,d,k,w)$.

Additionally, throughout this paper, when considering instances with $N \subseteq C$, we assume that $k < |N|$. This is without loss of generality. If $|N| \le k \le |C|$, we can choose $X=N$ as our final clustering. Then $\ell_i(X)=0$ for every agent $i \in N$, so this clustering trivially satisfies all the fairness properties we consider in this paper.

\paragraph{The (Droop) core and proportional fairness.}
Under the \emph{Hare quota}, any group of agents with total weight at least $1/k$ is collectively entitled to one center. A clustering $X$ is in the \emph{$\alpha$-core} (for $\alpha\ge 1$) if no entitled group can jointly $\alpha$-improve: there do not exist a center $c \in C$ and a set $S\subseteq N$ with $w(S) \ge 1/k$ such that
\[
  \alpha\cdot d(i,c) < \ell_i(X)\quad\text{for all }i\in S.
\]
We say such a pair $(S,c)$ is a \emph{blocking pair} and that the agents in $S$ \emph{$\alpha$-deviate} (from $X$) \emph{to} $c$.

Under the \emph{Droop quota}, a group of agents whose total weight exceeds $1/(k+1)$ is entitled to one center. The \emph{$\alpha$-Droop core} is defined analogously to the $\alpha$-core but with the Droop quota replacing the Hare quota. 
A clustering in the $\alpha$-Droop core is also in the $\alpha$-core.

\paragraph{NTU games and Scarf's theorem.}
A \emph{cooperative game with non-transferable utility} (an NTU game) $(N,V)$ consists of a set of \emph{agents} $N=[n]$ and a map $V$ that assigns to each nonempty coalition $\emptyset \ne S \subseteq N$ a \emph{payoff set} $V(S) \subseteq \mathbb{R}^{S}$ satisfying the following conditions:
\begin{enumerate}[label={(\arabic*)}]
	\item $V(S)$ is a nonempty, proper, and closed subset of $\mathbb{R}^S$,
	\item $V(S)$ is comprehensive, i.e., if $\mathbf{y} \in V(S)$ and $\mathbf{z} \le \mathbf{y}$ componentwise, then $\mathbf{z} \in V(S)$, and
	\item $V(S)$ is bounded above, i.e., there exists a constant $M$ such that $x_i \leq M$ for all $\mathbf{x} \in V(S)$ and $i \in S$.
\end{enumerate}

For each coalition of agents $S$, $V(S)$ represents the utility vectors that the members of $S$ can obtain on their own. For each such utility vector $\mathbf{y} \in V(S)$, $y_i$ is the utility obtained by $i \in S$. $V(N)$ represents the set of utility vectors that can be obtained by the \emph{grand coalition} (the entire group of agents) $N$. An outcome $\mathbf{x} \in V(N)$ is \emph{blocked} by a coalition $S$ if there exists $\mathbf{y} \in V(S)$ such that $y_i > x_i$ for every $i \in S$. An outcome not blocked by any coalition is said to be in the \emph{core}. The goal of an NTU game is to find a point in $V(N)$ that lies in the core. We call such points \emph{core points}. If at least one core point exists for a given NTU game $(N,V)$, then we say that the core of that game is \emph{nonempty}.

For a coalition $S$, let $\widetilde{V}(S) = V(S) \times \mathbb{R}^{N \setminus S}$. Intuitively, this represents all utility vectors in which the agents in $S$ obtain some utility vector from $V(S)$, while the agents in $N \setminus S$ receive arbitrary utilities. The set of core points in $V(N)$ is exactly

\[
V(N) \setminus \bigcup_{\emptyset \ne S \subseteq N} \operatorname{int}\widetilde{V}(S),
\]
where $\operatorname{int}\widetilde{V}(S)$ denotes the \emph{interior} of $\widetilde{V}(S)$. Due to the structure of $\widetilde{V}(S)$, a vector $\mathbf{x} \in \mathbb{R}^N$ belongs to $\operatorname{int}\widetilde{V}(S)$ if and only if there exists a vector $\mathbf{z} \in \widetilde{V}(S)$ such that $z_i > x_i$ for all $i \in S$.
Thus, it is quite easy to see that a vector $\mathbf{x} \in V(N)$ is blocked by a coalition $S$ if and only if it is in $\operatorname{int}\widetilde{V}(S)$.

Scarf's seminal theorem for NTU games \citep{scarf1967core} provides a sufficient condition for $V(N)$ to admit at least one utility vector in the core.

A family $\mathcal{S} \subseteq 2^N \setminus \{\emptyset\}$ of nonempty coalitions is called \emph{balanced} if there exist nonnegative weights $(\lambda_S)_{S \in \mathcal{S}}$ such that for each $i \in N$, we have $\sum_{S \in \mathcal{S} : i \in S}\lambda_S = 1$. An NTU game $(N,V)$ is \emph{balanced} if for every balanced family $\mathcal{S}$ of coalitions, $\bigcap_{S \in \mathcal{S}} \widetilde{V}(S) \subseteq V(N)$.

Scarf's theorem states the following:
\begin{theorem}[\citeauthor{scarf1967core},~\citeyear{scarf1967core}]\label{thm:scarf}
	A balanced NTU game has a nonempty core.
\end{theorem}

\section{The Existence of a 2-Droop Core Clustering}\label{sec:2-Droop-core}

We will now prove \Cref{thm:2-Droop-core} by showing how Scarf's theorem can be leveraged to establish the existence of a clustering in the $2$-Droop core for any proportionally fair clustering instance $I=(N,C,d,k,w)$ with $N \subseteq C$. We first construct an NTU game that corresponds to a fractional version of proportional clustering. We then apply Scarf's theorem to show that a fractional clustering in a newly defined fractional version of the core exists. Finally, we round this to an integral clustering using only agent centers, yielding a clustering in the $2$-Droop core.

Throughout the section, we assume that $C$ is finite for simplicity. We remove this assumption afterward.

\paragraph{Finding a fractional Droop core clustering.}\label{subsec:exact-core}
We define a \emph{fractional clustering} to be a vector $\mathbf{x} \in \mathbb{R}^C_{\ge 0}$ with $\mathbf{x}(C) \le 1$. $\mathbf{x}(C)=1$ will correspond to a clustering of size $k$. Note that an integral clustering $X \subseteq C$ with $|X| \leq k$ can be embedded into a fractional clustering by setting $x_c = \frac{1}{k}$ for each $c \in X$ and $x_c=0$ otherwise.

For a coalition $S$, let $\mathcal{X}(S)$ be the set of 
fractional clusterings with which $S$ may deviate:
\[
\mathcal{X}(S) = \left\{ \mathbf{x} \in \mathbb{R}^C_{\ge 0} : \mathbf{x}(C) \le w(S) \right\}.
\]

Intuitively, since a fractional clustering can have total mass $1$, any coalition of agents is allowed to deviate with its own fractional clustering whose total mass is proportional to its weight $w(S)$. Equivalently, since mass $\frac{1}{k}$ represents one center, a coalition of weight $\frac{1}{k}$ has enough mass to deviate with one center under the corresponding integral interpretation.

Let $\delta = \min \{ w(S) : S\subseteq N, w(S) > \frac{1}{k+1}\}$. Hence, a group of agents $S \subseteq N$ qualifies as a potential deviating coalition under the Droop quota if and only if $w(S) \ge \delta$.
Fix a constant $D > \max_{i \in N, c \in C}d(i,c)$. For each agent $i \in N$ and $\mathbf{x} \in \mathbb{R}^C$, let
\[
R_i(\mathbf{x}) = \min\left(\left\{ d(i,c): c \in C, \mathbf{x}(\{c' \in C: d(i,c') \le d(i,c)\}) \ge \delta \right\} \cup \{D\}\right).
\]
That is, if $\mathbf{x}$ is a fractional clustering with $\mathbf{x}(C) \ge \delta$, $R_i(\mathbf{x})$ denotes the smallest radius $r$ such that the ball centered at $i$ with radius $r$ includes at least $\delta$ of the mass of $\mathbf{x}$. If $\mathbf{x}(C) < \delta$, we let $R_i(\mathbf{x}) = D$, which by definition is larger than the distance of $i$ to any possible center. We treat $R_i$ as the fractional version of agent $i$'s loss function.

We will leverage Scarf's theorem to show the following:
\begin{lemma}\label{thm:fractional-core}
	There exists a fractional clustering $\mathbf{x}^* \in \mathbb{R}^C_{\ge 0}$ such that for all nonempty $S \subseteq N$ and $\mathbf{y} \in \mathcal{X}(S)$, $R_i(\mathbf{y}) \geq R_i(\mathbf{x}^*)$ for some $i \in S$.
\end{lemma}

The proof of the existence of $\mathbf{x}^*$ is the most technical step of this paper. We defer it to the end and first show that, assuming the existence of such a fractional clustering $\mathbf{x}^*$, the existence of an integral clustering in the $2$-Droop core follows quite easily.

\begin{lemma}\label{thm:rounding}
	For every weighted clustering instance $I=(N,C,d,k,w)$ with $N \subseteq C$ and finite $C$, there exists a clustering $X \subseteq N$ that is in the $2$-Droop core.
\end{lemma}
\begin{proof}
	Fix a clustering instance, and let $\mathbf{x}^*$ be a fractional clustering whose existence is guaranteed by \Cref{thm:fractional-core}. If $\mathbf{x}^*(C) < 1$, then pad $\mathbf{x}^*$ arbitrarily so that $\mathbf{x}^*(C)=1$. This can only weakly decrease $R_i(\mathbf{x}^*)$ for every $i \in N$ and does not affect the following proof.

	For each agent $i \in N$, let $B_i = \{c \in C: d(i,c) \le R_i(\mathbf{x}^*)\}$.
	Sort agents in increasing order of $R_i(\mathbf{x}^*)$; without loss of generality, let $R_1(\mathbf{x}^*) \le R_2(\mathbf{x}^*) \le \cdots \le R_n(\mathbf{x}^*)$.

	We will construct an integral clustering $X \subseteq N \subseteq C$ that consists solely of points in the metric space where an agent resides. We construct $X$ greedily using the following procedure.
	
	Start with $X \gets \emptyset$.
	In order of $i = 1, \ldots, n$, if $B_i \cap \left(\bigcup_{j \in X}B_j\right) = \emptyset$, then $X \gets X \cup \{i\}$.

	First, we show that $|X| \leq k$. This follows from the fact that $B_i$ and $B_j$ are disjoint for all distinct $i,j \in X$. By the definition of $B_i$, we have $\mathbf{x}^*(B_i) \geq \delta$. Thus,
	\begin{align*}
		|X| \cdot \delta &\le \sum_{i \in X} \mathbf{x}^*(B_i) = \mathbf{x}^*\left(\bigcup_{i \in X}B_i\right) \le \mathbf{x}^*(C) \le 1.
	\end{align*}
	Rearranging, we get $|X| \le {1}/{\delta} < k+1$.

	Next, we prove that $X$ is in the $2$-Droop core.
	First, we claim that $\ell_i(X) \le 2 R_i(\mathbf{x}^*)$ for all $i \in N$. If $i \in X$, then the result is trivial as $\ell_i(X)=0$ and $R_i(\mathbf{x}^*) \geq 0$. If $i \notin X$, then this means that there exists $j<i$ such that $j \in X$ and $B_i \cap B_j \ne \emptyset$. Take $c \in B_i \cap B_j$. Then,
	\begin{align*}
		d(i,j) \le d(i,c) + d(j,c) \le R_i(\mathbf{x}^*) + R_j(\mathbf{x}^*) \le 2 R_i(\mathbf{x}^*).
	\end{align*}
	
	To conclude that $X$ is in the $2$-Droop core, assume to the contrary that there exists a coalition $S \subseteq N$ and a center $c' \in C$ with $w(S) > \frac{1}{k+1}$ and $2 d(i,c') < \ell_i(X)$ for all $i \in S$. Then, by the definition of $\delta$, $w(S) \ge \delta$. Consider the fractional clustering $\mathbf{y}$ such that $y_{c'} = \delta$ and all other components receive mass $0$. Then $\mathbf{y}(C) = \delta \le w(S)$, so $\mathbf{y} \in \mathcal{X}(S)$. For each $i \in S$,
	\[
	R_i(\mathbf{y})=d(i,c')<\frac{1}{2}\ell_i(X)\le R_i(\mathbf{x}^*),
	\]
	which contradicts \Cref{thm:fractional-core}.
\end{proof}

\Cref{thm:fractional-core} and \Cref{thm:rounding} show that for any clustering instance with $N \subseteq C$ and $C$ finite, a clustering in the $2$-Droop core always exists. To prove \Cref{thm:2-Droop-core}, it remains to show that this still holds when $C$ is infinite.

\begin{proof}[Proof of \Cref{thm:2-Droop-core}]
	\Cref{thm:fractional-core} and \Cref{thm:rounding} establish that this theorem holds for every instance with a finite $C$.

	For the sake of contradiction, assume that there is a weighted instance $I=(N,C,d,k,w)$ with $|C| = \infty$ that admits no clustering $X \subseteq N$ in the $2$-Droop core.

	There are $\sum_{\ell=1}^k \binom{n}{\ell}$ subsets $X \subseteq N$ of size $1 \le |X| \le k$. For each $X\in \bigcup_{\ell=1}^k\binom{N}{\ell}$, choose a blocking pair $(S_X,c_X)$ with $S_X\subseteq N$ and $c_X\in C$ such that $w(S_X)>\frac{1}{k+1}$ and $2d(i,c_X)<\ell_i(X)$ for every $i\in S_X$. Let
	\[
		C'=N\cup\{c_X:X\in{\textstyle \bigcup_{\ell=1}^k}\tbinom{N}{\ell}\}.
	\]
	The set $C'$ is finite. Thus, the weighted instance $(N,C',d,k,w)$ must admit a clustering $X^* \subseteq N$ with $|X^*| \le k$ in the $2$-Droop core. However, $c_{X^*}\in C'$, and the pair $(S_{X^*},c_{X^*})$ blocks $X^*$, a contradiction.
\end{proof}

It remains to prove \Cref{thm:fractional-core} using Scarf's theorem.

\begin{proof}[Proof of \Cref{thm:fractional-core}]
	For each agent $i$ and $\mathbf{x} \in \mathbb{R}^C$, let $u_i(\mathbf{x}) = D - R_i(\mathbf{x})$. Since $R_i(\mathbf{x}) \in [0,D]$, we have $u_i(\mathbf{x}) \in [0, D]$.
	For each coalition $\emptyset \ne S \subseteq N$, let
	\[
	V(S) = \left\{ \mathbf{y} \in \mathbb{R}^S : \text{there exists } \mathbf{x} \in \mathcal{X}(S) \text{ such that } y_i \le u_i(\mathbf{x}), \ \forall i \in S \right\}.
	\]
	We claim that $(N,V)$ is a balanced NTU game. Fix a coalition $\emptyset \ne S \subseteq N$.
	The zero fractional clustering belongs to $\mathcal{X}(S)$ and satisfies $u_i(\mathbf{0})=0$ for every $i\in S$, so the zero vector belongs to $V(S)$. In particular, $V(S)$ is nonempty. Since $u_i(\mathbf{x}) \in [0,D]$, $V(S)$ is also proper and bounded above. For every $\mathbf{y} \in V(S)$, we have $y_i \leq D$ for all $i \in S$.
	To show that $V(S)$ is closed, observe that, since $C$ is finite, the definition of $R_i$ implies that
	\[
	R_i(\mathbf{x}) \in \{D\} \cup \{d(i,c) : c \in C\}.
	\]
	Therefore, each $u_i(\mathbf{x})=D-R_i(\mathbf{x})$ takes values in the finite set
	$U_i=\{0\}\cup\{D-d(i,c):c\in C\}$.
	Consequently, the set
	\[
	U_S=\left\{(u_i(\mathbf{x}))_{i\in S}:\mathbf{x}\in\mathcal{X}(S)\right\}
	\]
	of attainable utility profiles is finite. By the definition of $V(S)$,
	\[
	V(S)
	=
	\bigcup_{q\in U_S}
	\left\{
	\mathbf{y}\in\mathbb{R}^S:
	y_i\leq q_i \text{ for every } i\in S
		\right\}.
	\]
	For a fixed $q \in U_S$, the set

	\[\left\{
	\mathbf{y}\in\mathbb{R}^S:
	y_i\leq q_i \text{ for every } i\in S
	\right\}\]
	is clearly closed. Thus, the union over the finitely many vectors in $U_S$ is also closed. Therefore, $V(S)$ is closed.
	
	$V(S)$ is also comprehensive by definition. It remains to prove balancedness.
	Let $\mathcal{S}$ be a balanced family of coalitions with nonnegative weights $(\lambda_S)_{S \in \mathcal{S}}$. If $\mathbf{z} \in \bigcap_{S \in \mathcal{S}} \widetilde{V}(S)$, then for each $S \in \mathcal{S}$, take $\mathbf{x}^S \in \mathcal{X}(S)$ such that $z_i \le u_i(\mathbf{x}^S)$ for all $i \in S$. Let $\mathbf{x} = \sum_{S \in \mathcal{S}} \lambda_S \mathbf{x}^S$. Then,
	\begin{align*}
		\mathbf{x}(C) &= \sum_{S \in \mathcal{S}} \lambda_S \mathbf{x}^S(C) \\
		&\leq \sum_{S \in \mathcal{S}} \lambda_S w(S) \\
		&= \sum_{S \in \mathcal{S}} \lambda_S \sum_{i \in S} w(i) \\
		&= \sum_{i \in N} w(i) \sum_{S \in \mathcal{S} : i \in S} \lambda_S \\
		&= \sum_{i \in N} w(i) \\
		&= w(N) = 1.
	\end{align*}
	The inequality follows from the definition of $\mathcal{X}(S)$. Hence, $\mathbf{x} \in \mathcal{X}(N)$. We now claim that $z_i \le u_i(\mathbf{x})$ for every $i \in N$. Fix $i \in N$. Since $u_i(\mathbf{x}) \in [0,D]$, the claim is trivial if $z_i \le 0$. Assume $z_i > 0$. For any $S \in \mathcal{S}$ with $i \in S$, since $z_i \le u_i(\mathbf{x}^S) = D - R_i(\mathbf{x}^S)$, we have $R_i(\mathbf{x}^S) \le D - z_i < D$. That is, $\mathbf{x}^S(\{c \in C: d(i,c) \le D - z_i\}) \ge \delta$. Hence,
	\begin{align*}
		\mathbf{x}(\{c \in C: d(i,c) \le D - z_i\}) &= \sum_{S \in \mathcal{S}} \lambda_S \mathbf{x}^S (\{c \in C: d(i,c) \le D - z_i\}) \\
		& \ge \sum_{S \in \mathcal{S}: i \in S} \lambda_S \delta \\
		&= \delta.
	\end{align*}
	That is, $R_i(\mathbf{x}) \le D - z_i$, i.e., $u_i(\mathbf{x}) \ge z_i$. Hence, $\mathbf{z} \in V(N)$ and the game is balanced.
	
	Applying Scarf's theorem (\Cref{thm:scarf}), there exists $\mathbf{z}^* \in V(N) \setminus \bigcup_{\emptyset \ne S \subseteq N} \operatorname{int}\widetilde{V}(S)$. By the definition of $V(N)$, there must exist some $\mathbf{x}^* \in \mathcal{X}(N)$ such that $z^*_i \le u_i(\mathbf{x}^*)$ for all $i \in N$.
	We claim that $\mathbf{x}^*$ satisfies the requirement of the lemma. Assume to the contrary that there exists $(S,\mathbf{y})$ with nonempty $S \subseteq N$, $\mathbf{y} \in \mathcal{X}(S)$, and $R_i(\mathbf{y}) < R_i(\mathbf{x}^*)$ for every $i \in S$. Then by definition, $u_i(\mathbf{y}) > u_i(\mathbf{x}^*) \ge z_i^*$ for all $i \in S$. Let $\epsilon = \frac{1}{2} \min_{i \in S} (u_i(\mathbf{y}) - z_i^*) > 0$. Then the vector $\mathbf{z}'$ defined by $z'_i = z^*_i + \epsilon$ for all $i \in S$ satisfies $z'_i < u_i(\mathbf{y})$ for every $i \in S$. Hence, we have $\mathbf{z}' \in V(S)$. Since $z'_i$ is strictly larger than $z^*_i$ on every coordinate $i \in S$, the characterization of $\operatorname{int}\widetilde{V}(S)$ above contradicts the fact that $\mathbf{z}^* \not\in \operatorname{int}\widetilde{V}(S)$.
	\end{proof}

\section{Discussion}

\paragraph{Further generalizations.} While \Cref{thm:2-Droop-core} completely settles the question of the optimal core bounds for the $N \subseteq C$ special case of proportionally fair clustering, several interesting generalizations of the model remain open. One of the most compelling is the sortition model of \citet{ebadian2024boosting}. In this model, we have $N = C$, and the goal is to find a lottery over clusterings that provides a core-approximation guarantee for every realized clustering while ensuring that each agent is selected with probability $\nicefrac{k}{|N|}$. Interestingly, \citet{ebadian2024boosting} showed a lower bound of $2$ on the ex post core-approximation factor for any such lottery. The techniques used in this paper do not appear to extend easily to establishing a lottery over clusterings that makes this lower bound tight.

\paragraph{Separation between $N \subseteq C$ and the general case.} While the $N \subseteq C$ case is now fully resolved, the general case remains open. As stated previously, the recent lower bounds of \citet{cookson2026improved} place the known bounds for the general case at $[2.1508, 1 + \sqrt{2}]$. This is especially interesting because it establishes a separation between the general case and the $N \subseteq C$ case. The condition $N \subseteq C$ provides enough structure to the instance that a $2$-approximation is always possible, while this is not the case in the general setting. Resolving the gap in the general setting appears to be the largest theoretical open problem remaining in the study of the core in clustering.

\paragraph{Polynomial-time computation.} \Cref{thm:2-Droop-core} shows the existence of a clustering in the $2$-Droop core. It is easy to see that, at least when the candidate set is finite, if we are given the fractional clustering vector $\mathbf{x}^*$ whose existence is established using Scarf's theorem, then $\mathbf{x}^*$ can be rounded to the final integral clustering in polynomial time. However, it is known that, in general, finding the utility vector implied by Scarf's theorem is PPAD-complete \citep{kintali2013reducibility}. Thus, this does not give a clear path to determining whether the fractional clustering $\mathbf{x}^*$ can be computed in polynomial time in this setting.

\bibliography{refs}

\end{document}